\documentclass[11pt, letterpaper]{article}
\usepackage{fullpage}
\usepackage{times}
\usepackage{amsmath,amsfonts,amssymb,amsthm}

\usepackage{mathptmx}

\usepackage[text={6.5in,8.8in}]{geometry}

\usepackage[compact]{titlesec}
\titlespacing{\section}{0pt}{2.5ex}{1.25ex}
\titlespacing{\subsection}{0pt}{2ex}{0.8ex}
\titlespacing{\subsubsection}{0pt}{0.5ex}{0ex}

\usepackage[colorlinks=true,linkcolor=red,citecolor=blue,urlcolor=red]{hyperref}
\usepackage{framed}
\usepackage{verbatim}

\begin{document}

\newtheorem{fact}{Fact}[section]
\newtheorem{rules}{Rule}[section]
\newtheorem{conjecture}{Conjecture}[section]
\newtheorem{theorem}{Theorem}[section]
\newtheorem{hypothesis}{Hypothesis}
\newtheorem{remark}{Remark}
\newtheorem{proposition}{Proposition}
\newtheorem{corollary}{Corollary}[section]
\newtheorem{lemma}{Lemma}[section]
\newtheorem{claim}{Claim}
\newtheorem{definition}{Definition}[section]

\newenvironment{proofof}[1]{\smallskip
\noindent {\bf Proof of #1.  }}{\hfill$\Box$
\smallskip}

\newenvironment{reminder}[1]{\smallskip

\noindent {\bf Reminder of #1  }\em}{
}

\def \BP {{\sf BP}}
\def \coRP {{\sf coRP}}
\def \EXACT {{\sf EXACT}}
\def \SYM {{\sf SYM}}
\def \SAC {{\sf SAC}}
\def \SUBEXP {{\sf SUBEXP}}
\def \ZPSUBEXP {{\sf ZPSUBEXP}}
\def \SYMACC {{\sf SYM\text{-}ACC} }
\def \QED {{\hfill$\Box$}}
\def \PH {{\sf PH}}
\def \RP {{\sf RP}}
\def \coMA {{\sf coMA}}
\def \ACC {{\sf ACC}}
\def \SYM {{\sf SYM}}
\def \coNP {{\sf coNP}}
\def \BPP {{\sf BPP}}
\def \NC {{\sf NC}}
\def \ZPE {{\sf ZPE}}
\def \NE {{\sf NE}}
\def \E {{\sf E}}
\def \poly { \text{\rm poly} }
\def \TC {{\sf TC}}
\def \DTS {{\sf DTS}}
\def \R {{\mathbb R}}
\def \Z {{\mathbb Z}}
\def \P {{\sf P}}
\def \MA {{\sf MA}}
\def \AM {{\sf AM}}
\def \MATIME {{\sf MATIME}}
\def \QP {{\sf QP}}
\def \coNQP {{\sf coNQP}}
\def \NP {{\sf NP}}
\def \EXP {{\sf EXP}}
\def \NTISP {{\sf NTISP}}
\def \DTISP {{\sf DTISP}}
\def \TISP {{\sf TISP}}
\def \T {{\sf TIME}}
\def \TH {{\rm TH}}
\def \NEAR {{\rm NEAR}}
\def \TIME {\T}

\def \AC {{\sf AC}}
\def \BPTIME {{\sf BPTIME}}
\def \SPACE {{\sf SPACE}}
\def \RE {{\sf RE}}
\def \ZPSUBEXP {{\sf ZPSUBEXP}}
\def \coREXP {{\sf coREXP}}
\def \NQL {{\sf NQL}}
\def \QL {{\sf QL}}
\def \RTIME {{\sf RTIME}}
\def \NSUBEXP {{\sf NSUBEXP}}
\def \MAEXP {{\sf MAEXP}}
\def \PP {{\sf PP}}
\def \PSPACE {{\sf PSPACE}}
\def \NT {{\sf NTIME}}
\def \N {{\mathbb N}}
\def \NTIME {\NT}

\def \ATIME {\AT}

\def \Z {{\mathbb Z}}
\def \F {{\mathbb F}}
\def \isin {\subseteq}
\def \isnotin {\nsubseteq}
\def \eps {\varepsilon}
\def\polylog{\operatorname{polylog}}

\newcommand{\card}[1]{\ensuremath{\left|#1\right|}}
\newcommand{\ip}[2]{\ensuremath{\left<#1,#2\right>}}
\newcommand{\brac}[1]{\ensuremath{\text{\bf [}#1\text{\bf ]}}}
\newcommand{\mv}[2]{\ensuremath{\mathbf{MV}\!\left(#1,#2\right)}}
\newcommand{\hp}[1]{\text{\bf HOPCROFT}_{#1}}
\newcommand{\D}{\ensuremath{\textrm{deg}}}

\title{Probabilistic Polynomials and Hamming Nearest Neighbors\\ (Full Version)}
\author{
Josh Alman\footnote{Computer Science Department, Stanford University. Supported by NSF CCF-1212372 and NSF DGE-114747} \and
Ryan Williams\footnote{Computer Science Department, Stanford University, {\tt rrw@cs.stanford.edu}. Supported in part by a David Morgenthaler II Faculty Fellowship, and NSF CCF-1212372. Any opinions, findings and conclusions or recommendations expressed in this material are those of the authors and do not necessarily reflect the views of the National Science Foundation.} 
}

\maketitle

\begin{abstract} We show how to compute any symmetric Boolean function on $n$ variables over any field (as well as the integers) with a probabilistic polynomial of degree $O(\sqrt{n \log(1/\eps)})$ and error at most $\eps$. 
 The degree dependence on $n$ and $\eps$ is optimal, matching a lower bound of Razborov (1987) and Smolensky (1987) for the MAJORITY function. The proof is constructive: a low-degree polynomial can be efficiently sampled from the distribution.

This polynomial construction is combined with other algebraic ideas to give the first subquadratic time algorithm for computing a (worst-case) batch of Hamming distances in superlogarithmic dimensions, \emph{exactly}. To illustrate, let $c(n) : \N \rightarrow \N$. Suppose we are given a database $D$ of $n$ vectors in $\{0,1\}^{c(n) \log n}$ and a collection of $n$ query vectors $Q$ in the same dimension. For all $u \in Q$, we wish to compute a $v \in D$ with minimum Hamming distance from $u$. We solve this problem in $n^{2-1/O(c(n) \log^2 c(n))}$ randomized time. Hence, the problem is in ``truly subquadratic'' time for $O(\log n)$ dimensions, and in subquadratic time for $d = o((\log^2 n)/(\log \log n)^2)$. We apply the algorithm to computing pairs with maximum inner product, closest pair in $\ell_1$ for vectors with bounded integer entries, and pairs with maximum Jaccard coefficients.


\end{abstract}

\thispagestyle{empty}
\newpage
\setcounter{page}{1}

\section{Introduction}

Recall the \emph{Hamming nearest neighbor problem} (HNN): given a set $D$ of $n$ database points in the $d$-dimensional hypercube, we wish to preprocess $D$ to support queries of the form $q \in \{0,1\}^d$, where a query answer is a point $u \in D$ that differs from $q$ in a minimum number of coordinates. Minsky and Papert~(\cite{perceptrons}, Chapter 12.7) called this the ``Best Match'' problem, and it has been widely studied since. Like many situations where one wants to find points that are ``most similar'' to query points, HNN is fundamental to modern computing, especially in search and error correction~\cite{IndykSurvey}. However, known exact solutions to the problem require a data structure of $2^{\Omega(d)}$ size (storing all possible queries) or query time $\Omega(n/\poly(\log n))$ (trying nearly all the points in the database). This is one of many examples of the \emph{curse of dimensionality} phenomenon in search, with corresponding data structure lower bounds. For instance, Barkol and Rabani \cite{BarkolRabani} show a size-query tradeoff for HNN in $d$ dimensions in the cell-probe model: if one uses $s$ cells of size $b$ to store the database and probes at most $t$ cells in a query, then either $s = 2^{\Omega(d/t)}$ or $b = n^{\Omega(1)}/t$. 

During the late 90's, a new direction opened in the search for better nearest neighbor algorithms. The driving intuition was that it may be easier to find and generally good enough to have \emph{approximate} solutions: points with distance within $(1+\eps)$ of the optimum. Utilizing novel hashing and dimensionality reduction techniques, this beautiful line of work has had enormous impact~\cite{Kleinberg, IndykMotwani, KOR, Panigrahy, AndoniIndyk, GregLightbulb, AndoniBeyond, andoni2015optimal}. Still, when turning to approximations, the exponential-in-$d$ dependence generally turns into an exponential-in-$1/\eps$ dependence, leading to a ``curse of approximation''~\cite{PatrascuThesis}, with lower bounds matching this intuition \cite{CCGL, CR, AIP}. For example, Andoni, Indyk, and Patrascu~\cite{AIP} prove that any data structure for $(1+\eps)$-approximate HNN using $O(1)$ probes requires $n^{\Omega(1/\eps^2)}$ space.

In this paper, we revisit exact nearest neighbors in the Hamming metric. We study the natural off-line problem of answering $n$ Hamming nearest neighbor queries at once, on a database of size $n$. We call this the {\sc Batch Hamming Nearest Neighbor} problem (BHNN). Here the aforementioned data structure lower bounds no longer apply---there is no information bottleneck. Nevertheless, known algorithms for BHNN still run in either about $n^2 d^{\Omega(1)}$ time (try all pairs)~\cite{GumLipton, MinKerui} or about $n 2^{\Omega(d)}$ time (build a table of all possible query answers). We improve over both these bounds for $\log n \leq d \leq o(\log^2 n/\log \log n)$. Our approach builds on a recently developed framework ~\cite{WilliamsAPSP14,WilliamsFSTTCS14,AbboudWY15}. In this work, the authors show how several famous stubborn problems can yield faster algorithms, by constructing low-complexity circuits for solving simple repeated subparts of the problem. The overall strategy is to convert the simple repeated pieces into polynomials of a special form, then to evaluate the polynomials on many points fast, via an algebraic matrix multiplication.

For the problems considered in earlier work, these polynomials can be constructed using 30-year-old ideas from circuit complexity. More formally, if $f$ is a Boolean function on $n$ variables and $R$ is a ring, a \emph{probabilistic polynomial over $R$ for $f$ with error $\eps$ and degree $d$} is a distribution ${\cal D}$ of degree-$d$ polynomials over $R$ with the property that for all $x \in \{0,1\}^n$, $\Pr_{p \sim {\cal D}}[p(x) = f(x)] \geq 1-\eps$. 
Razborov~\cite{Razborov} and Smolensky~\cite{Smolensky87} showed how to construct low-degree probabilistic polynomials for every $f$ computable by a small constant-depth circuit composed of PARITY, AND, and OR gates. They also proved that probabilistic polynomials for MAJORITY with constant error require $\Omega(\sqrt{n})$ degree, concluding circuit lower bounds for MAJORITY. Earlier papers~\cite{WilliamsAPSP14,WilliamsFSTTCS14,AbboudWY15} used this low-degree construction to derive faster algorithms for problems such as dense all-pairs shortest paths, longest common substring with wildcards, and batch partial match queries. 

Developing a faster algorithm for computing Hamming nearest neighbors requires more care than prior work. In the setting of this paper, the ``repeated'' computation we need to consider is that of finding a pair of vectors among a small set which have small Hamming distance. But computing Hamming distance requires \emph{counting} bits, which means we are implicitly computing a MAJORITY of some kind. This is fundamentally harder than the constant-depth computations handled in prior work. Proceeding anyway, we prove in this paper that the Razborov-Smolensky $\sqrt{n}$ lower bound is tight up to constant factors: there is a probabilistic polynomial for MAJORITY achieving degree  $O(\sqrt{n})$ with constant error. In fact, we show that this degree can be achieved for any symmetric Boolean function. We use this to get a subquadratic time algorithm for Hamming distance computations up to about $\log^2 n$ dimensions.

\subsection{Our Results}

Recently, Srinivasan~\cite{Srinivasan13} gave a probabilistic polynomial for the MAJORITY function of degree $\sqrt{n \log(1/\eps)} \cdot \polylog (n)$ over any field. We construct a probabilistic polynomial for MAJORITY on $n$ variables with optimal dependence on $n$ and error $\eps$ over any field or the integers.

\begin{theorem} \label{MAJpoly} Let $R$ be a field, or the integers. There is a probabilistic polynomial over $R$ for MAJORITY on $n$ variables with error $\eps$ and degree $d(n,\eps)=O(\sqrt{n \log(1/\eps)})$. Furthermore, a polynomial can be sampled from the probabilistic polynomial distribution in $\tilde{O}(\sum_{i=0}^{d(n,\eps)} {n \choose i})$ time.
\end{theorem}

As mentioned above, Razborov and Smolensky's famous lower bounds for MAJORITY implies a degree lower bound of precisely $\Omega(\sqrt{n})$ in the case of constant $\eps$. For non-constant $\eps$, an asymptotically lower-degree polynomial for MAJORITY (in either $\eps$ or $n$) could be used to compute the majority of $\log(1/\epsilon)$ bits with $o(\log(1/\epsilon))$ degree and error $\eps$, which is impossible---the exact degree of MAJORITY on $n$ bits equals $n$, over any field and $\Z$. Theorem~\ref{MAJpoly} can also be applied to derive $O(\sqrt{n \log(1/\eps)})$ degree probabilistic polynomials for \emph{every} symmetric function (again improving on Srinivasan~\cite{Srinivasan13}). 

\begin{theorem} \label{SYMpoly} Let $R$ be a field, or the integers. There is a probabilistic polynomial over $R$ for any symmetric Boolean function on $n$ variables with error $\eps$ and degree $d(n,\eps)=O(\sqrt{n \log(1/\eps)})$.
\end{theorem}

We use Theorem~\ref{MAJpoly} to derive several new algorithms\footnote{We stress that the polynomials of \cite{Srinivasan13} do not seem to imply the algorithms of this paper; removing the extra polylogarithmic factor is important!}. The main application is a solution to the BHNN problem mentioned earlier, where we are given $n$ query points and an $n$-point database, and wish to answer all $n$ Hamming distance queries in one shot. We show:

\begin{theorem} \label{HammingNN} Let $D \subseteq \{0,1\}^{c \log n}$ be a database of $n$ vectors, where $c$ can be a function of $n$. Any batch of $n$ Hamming nearest neighbor queries on $D$ can be answered in randomized $n^{2-1/O(c \log^2 c)}$ time, whp. \end{theorem}

For instance, if $d = O(\log n)$, then the algorithm runs in \emph{truly subquadratic} time: $n^{2-\eps}$, for some  $\eps > 0$. To our knowledge, this is the first known improvement over $n^2$ time for the case where $d \geq \log n$. In general, our algorithm improves over $n^2$ for dimensions up to $o(\log^2 n/(\log \log n)^2)$.\footnote{The logarithmic decrease in degree compared to previous results in Theorem \ref{MAJpoly} is crucial for achieving this truly subquadratic runtime: the resulting decrease in the number of monomials in Theorem \ref{Hammingprobpoly} will be necessary to get the runtime in Theorem \ref{BichromaticHD} of our algorithm's analysis.}

Theorem~\ref{HammingNN} follows from a similar running time for {\sc Bichromatic Hamming Closest Pair}: given  $k$ and a collection of ``red'' and ``blue'' Boolean vectors, determine if there is a red and blue vector with Hamming distance at most $k$. Such bichromatic problems are central to algorithms over metric spaces.

The versatility of the Hamming metric makes Theorem~\ref{HammingNN} highly applicable. For example, we can also solve closest pair in $\ell_1$ norm with bounded integer entries, as well as {\sc Bichromatic Min Inner Product}: given an integer $k$ and a collection of red and blue Boolean vectors, determine if there is a red and blue vector with inner product at most $k$. We show that these problems are in $n^{2-1/O(c \log^2 c)}$ randomized time, by simple reductions (Theorem~\ref{L1NN} and Theorem~\ref{BichromaticMIP}). As a consequence, closest pair problems in other measures, such as the Jaccard distance, can also be solved in subquadratic time.



It is important to keep in mind that sufficiently fast off-line Hamming closest pair algorithms would yield a breakthrough in satisfiability algorithms, so there is a potential limit:

\begin{theorem} \label{SETH} Suppose there is $\eps > 0$ such that for all constant $c$, {\sc Bichromatic Hamming Closest Pair} can be solved in $2^{o(d)} \cdot n^{2-\eps}$ time on a set of $n$ points in $\{0,1\}^{c \log n}$. Then the Strong Exponential Time Hypothesis is false.
\end{theorem}

The proof is actually a reduction from the (harder-looking) {\sc Orthogonal Vectors} problem, where it is well-known that $n^{2-\eps}$ time would refute SETH~\cite{Williams05}. For completeness, the proof is in Section~\ref{appendix-SETH}.

\subsection{Other Related Work}

The ``planted'' case of Hamming distance has been studied extensively in learning theory and cryptography. In this setting, all vectors are chosen uniformly at random, except for a planted pair of vectors with Hamming distance much smaller than the expected distance between two random vectors. Two recent references are notable: G. Valiant~\cite{GregLightbulb} gave a breakthrough $O(n^{1.62})$ time algorithm, which is \emph{independent} of the vector dimension and the Hamming distance of the planted pair. Valiant also gives a $(1+\eps)$-approximation to the closest pair problem in Hamming distance running in $n^{2-\Omega(\sqrt{\eps})}$ time. See \cite{EUROCRYPTNN} for very recent work on batch Hamming distance computations in cryptoanalysis.

Gum and Lipton \cite{GumLipton} observe that $n^2$ Hamming distances can be computed in $O(n^2 d^{0.4})$ time via a direct application of fast matrix multiplication. An extension to arbitrary alphabets was obtained by \cite{MinKerui}. For our situation of interest ($d \ll n$) this is only a minor improvement over the $O(n^2 d)$ cost of the obvious algorithm.

\section{Preliminaries}

We assume basic familiarity with algorithms, complexity theory, and properties of polynomials. It is worth noting that for a weaker notion of approximation, it is not hard to construct low-degree polynomials that \emph{correlate} well with MAJORITY, and in fact any symmetric function. In particular, for every symmetric function and $\eps > 0$ there is a single degree-$O(\sqrt{n})$ polynomial that agrees with the function on at least $1-\eps$ of the points in $\{0,1\}^n$: take a polynomial that outputs the symmetric function's value on the inputs of Hamming weight $[n/2-\Omega(\sqrt{n}),n/2+O(\sqrt{n})]$. A constant fraction of the $n$-bit inputs are in this interval, and polynomial interpolation yields an $O(\sqrt{n})$-degree polynomial. (See Lemma~\ref{Apoly}.) Our situation is more difficult: we want \emph{all} inputs to have a high chance of agreement with our symmetric function, when we sample a polynomial.

We need one lemma from prior work on efficiently evaluating polynomials over a combinatorial rectangle of inputs. The lemma was proved and used in earlier work~\cite{WilliamsAPSP14, AbboudWY15} to design randomized algorithms for many problems. 

\begin{lemma}[\cite{WilliamsAPSP14}] \label{polyeval}
Given a polynomial $P(x_1,\ldots,x_d,y_1,\ldots,y_d)$ over a (fixed) finite field with at most $n^{0.17}$ monomials, and two sets of $n$ inputs $A = \{a_1,\ldots,a_n\} \subseteq \{0,1\}^d$, $B = \{b_1,\ldots,b_n\} \subseteq \{0,1\}^d$, we can evaluate $P$ on all pairs $(a_i,b_j) \in A \times B$ in $\tilde{O}(n^2 + d \cdot n^{1.17})$ time. 
\end{lemma}

At the heart of Lemma~\ref{polyeval} is a rectangular (but not necessarily impractical!) matrix multiplication algorithm. For more details, see the references.

\subsection{Notation}
In what follows, for $(x_1,\ldots,x_n) \in \{0,1\}^n$ define $|x| := \sum_{i=1}^n x_i$. For a logical predicate $P$, we use the notation $\brac{P}$ to denote the function which outputs $1$ when $P$ is true, and $0$ when $P$ is false.

For $\theta \in [0,1]$, define $\TH_\theta : \{ 0,1 \}^n \to \{ 0, 1 \}$ to be the \emph{threshold function} $\TH_\theta(x_1,\ldots,x_n) := \brac{|x|/n \geq \theta}$. In particular, $\TH_{1/2} = {\rm MAJORITY}$. We also define $\NEAR_{\theta, \delta} : \{ 0,1\}^n \to \{0,1\}$, such that $\NEAR_{\theta, \delta}(x) := \brac{|x|/n \in [\theta - \delta, \theta + \delta]}$. Intuitively, $\NEAR_{\theta, \delta}$ checks whether $|x|/n$ is ``near'' $\theta$, with error $\delta$. 

\section{Probabilistic Polynomial for MAJORITY: Proof of Theorem~\ref{MAJpoly}} 

In this section, we prove Theorem~\ref{MAJpoly}. To do so, we construct a probabilistic polynomial for $\TH_\theta$ over $\Z[x_1, \ldots, x_n]$ which has degree $O(\sqrt{n \log(1/\epsilon)})$ and on each input is correct with probability at least $1 - \epsilon$. 

\paragraph{Intuition for the construction.} First, let us suppose $|x|/n$ is not too close to $\theta$: in particular $|x|/n$ is not within $\delta = O(\sqrt{\log(1/\epsilon) / n})$ of $\theta$. Then, if we construct a new smaller vector $\tilde{x}$ by sampling $1/10$ of the entries of $x$, it is likely that $|\tilde{x}|/(n/10)$ lies on the same side of $\theta$ as $|x|/n$. This suggests a \emph{recursive} strategy: we can use our polynomial construction on the sample $\tilde{x}$. Second, if $|x|/n$ is close to $\theta$, then by interpolating, we can use an exact polynomial of degree $O(\sqrt{n \log(1/\epsilon)})$ (which we call $A_{n,\theta,g}$) that is guaranteed to give the correct answer. To decide which of the two cases we are in, we will use a probabilistic polynomial for $\NEAR$ (on a smaller number of variables), which can itself be written as the product of two probabilistic polynomials for $\TH$. The degree incurred by recursive calls can be adjusted to have tiny overhead, with the right parameters.

In comparison, Srinivasan~\cite{Srinivasan13} takes a number theoretic approach. For $\Omega(\log n)$ different primes $p$, his polynomial uses $p-1$ probabilistic polynomials in order to determine the Hamming weight of the input$\pmod{p}$. Then, it uses an exact polynomial inspired by the Chinese Remainder Theorem to determine the true Hamming weight of the input, and whether it is at least $n/2$. This approach works on a more general class of functions than ours, called $W$-sum determined, which are determined by a weighted sum of the input coordinates. However, the number of primes being considered inherently means that this type of approach will incur extra logarithmic degree increases. In fact, we also give a better probabilistic degree for every symmetric function.

\paragraph{Interpolating Polynomial} Let $A_{n, \theta, g} : \{0,1\}^n \to \Z$ be an exact polynomial of degree at most $2 g \sqrt{n} + 1$ which gives the correct answer to $\TH_\theta$ for any vector $x$ with $|x| \in [\theta n - g \sqrt{n}, \theta n + g \sqrt{n}]$, and can give arbitrary answers to other vectors. Such a polynomial $A_{n, \theta, g}$ can be derived from prior work (at least over fields~\cite{Srinivasan13}), but for completeness, we nonetheless prove its existence.\footnote{It is not immediately obvious from univariate polynomial interpolation that $A_{n, \theta, g}$ exists as described, since the univariate polynomial $p$ such that $A_{n, \theta, g}(x) = p(|x|)$ typically has rational (non-integer) coefficients.} 






\begin{lemma} \label{Apoly}
For any integers $n, r, k$ with $n \geq k+r$ and any integers $c_1, \ldots, c_r$, there is a multivariate polynomial $p : \{0,1\}^n \to \Z$ of degree $r-1$ with integer coefficients such that $p(x) = c_i$ for all $x \in \{ 0,1\}^n$ with Hamming weight $|x| = k + i$.
\end{lemma}

Lemma \ref{Apoly} is more general than a result claimed without proof by Srinivasan (\cite{Srinivasan13}, Lemma 14). It also generalizes of a theorem of Bhatnagar et al.~(\cite{BhatnagarGL06}, Theorem 2.8). 

\begin{proof}
Our polynomial $p$ will have the form \[p(x_1, \ldots, x_n) = \sum_{i=0}^{r-1} a_i \cdot \sum_{\substack{\alpha \in \{ 0,1\}^n \\ |\alpha| = i}} \left( \prod_{j=1}^{n} x_j^{\alpha_j} \right)\] for some constants $a_0, \ldots, a_{r-1}$. Hence, we will get that for any $x \in \{ 0, 1 \}^n$:
\[p(x) = \sum_{i=0}^{r-1} \binom{|x|}{i} a_i.\] Define the matrix:
\[M = 
\left( \begin{matrix}
  \binom{k+1}{0} & \binom{k+1}{1} & \cdots & \binom{k+1}{r-1} \\
  \binom{k+2}{0} & \binom{k+2}{1} & \cdots & \binom{k+2}{r-1} \\
  \vdots & \vdots & \ddots & \vdots \\
  \binom{k+r}{0} & \binom{k+r}{1} & \cdots & \binom{k+r}{r-1}
 \end{matrix} \right).\]
The conditions of the stated lemma are that $$M 
\left( \begin{matrix}
  a_0 \\
  a_1 \\
  \vdots \\
  a_{r-1}
 \end{matrix} \right)
=
\left( \begin{matrix}
  c_1 \\
  c_2 \\
  \vdots \\
  c_r
 \end{matrix} \right).$$ 
By Lemma \ref{det} (proved below), $M$ always has determinant 1. Because $M$ is a matrix with integer entries and determinant 1, its inverse $M^{-1}$ is also an integer matrix. Multiplying through by $M^{-1}$ above gives integer expressions for the $a_i$, as desired.
\end{proof}

\begin{lemma} \label{det}
For any univariate polynomials $p_1, p_2, \ldots, p_{r}$ such that $p_i$ has degree $i-1$, and any pairwise distinct $x_1, x_2, \ldots, x_{r} \in \Z$, the matrix $$M = 
\left( \begin{matrix}
  p_1(x_1) & p_2(x_1) & \cdots & p_{r}(x_1) \\
  p_1(x_2) & p_2(x_2) & \cdots & p_{r}(x_2) \\
  \vdots & \vdots & \ddots & \vdots \\
  p_1(x_{r}) & p_2(x_{r}) & \cdots & p_{r}(x_{r}) \\
 \end{matrix} \right)$$ has determinant
\[det(M) = \left( \prod_{i=1}^{r} c_i \right) \cdot \left( \prod_{1 \leq i < j \leq r} (x_j - x_i) \right),\] where $c_i$ is the coefficient of $x^{i-1}$ in $p_i$.
\end{lemma}

\begin{proof}
For $i$ from $1$ up to $r-1$, we can add multiples of column $i$ of $M$ to the subsequent columns in order to make the coefficient of $x^{i-1}$ in all the other columns $0$. The resulting matrix is

$$M' = 
\left( \begin{matrix}
  c_1 & c_2 x_1 & \cdots & c_{r} x_1^{r-1} \\
  c_1 & c_2 x_2  & \cdots & c_{r} x_2^{r-1} \\
  \vdots & \vdots & \ddots & \vdots \\
  c_1 & c_2 x_{r} & \cdots & c_{r} x_{r}^{r-1} \\
 \end{matrix} \right).$$

This is a Vandermonde matrix which has the desired determinant. \end{proof}

\paragraph{Definition.} Let $n$ be an integer for which we want to compute $TH_\theta$. Let $M_{m, \theta, \epsilon} : \{ 0,1\}^m \to \Z$ denote the probabilistic polynomial for $\TH_\theta$ with error $\leq \epsilon$ degree as described above for all $m < n$. We can assume as a base case that when $m$ is constant, we simply use the exact polynomial for $\TH_\theta$.

Define \[S_{m, \theta, \delta, \epsilon}(x) := (1-M_{m, \theta + \delta, \epsilon}(x)) \cdot M_{m, \theta - \delta, \epsilon}(x).\]
Assuming $M_{n, \theta, \epsilon}$ works as prescribed (with $\leq \epsilon$ error), this is a probabilistic polynomial for $\NEAR_{\theta, \delta}$ with error at most $2 \epsilon$. For $x \in \{0,1\}^n$, let $\tilde{x} \in \{0,1\}^{n/10}$ be a vector of length $n/10$, where each entry is an independent and uniformly random entry of $x$. Hence, each entry of $\tilde{x}$ is a probabilistic polynomial in $x$ of degree $1$. Let $a = \sqrt{10}\cdot  \sqrt{\ln (1/\epsilon)}$. Our probabilistic polynomial for $\TH_\theta$ on $n$ variables is defined to be:
\[M_{n, \theta, \epsilon}(x) := A_{n, \theta, 2a}(x) \cdot S_{n/10, \theta, a/\sqrt{n}, \epsilon/4}(\tilde{x})  +  M_{n/10, \theta, \epsilon/4}(\tilde{x}) \cdot (1 - S_{n/10, \theta, a/\sqrt{n}, \epsilon/4}(\tilde{x})).\]

Note that $\tilde{x}$ denotes the \emph{same} randomly chosen vector in each of its appearances, and $S_{n/10, \theta, a/\sqrt{n}, \epsilon/4}$ denotes the same draw from the random polynomial distribution in both of its appearances.

\paragraph{Degree of $M_{n, \theta, \eps}$.} First we show by induction on $n$ that $M_{n, \theta, \eps}$ has degree $\leq 41 \sqrt{n \ln(1/\epsilon)}$. Assume that $M_{m, \theta, \epsilon}$ has degree $\leq 41 \sqrt{m \ln(1/\epsilon)}$ for all $m < n$. We have:
\begin{align*} \deg (M_{n, \theta, \epsilon}) &= \max \left\{ \deg\left[A_{n, \theta, 2a}(x) \cdot S_{n/10, \theta, a/\sqrt{n}, \epsilon/4}(\tilde{x}) \right], \deg\left[M_{n/10, \theta, \epsilon/4}(\tilde{x}) \cdot (1 - S_{n/10, \theta, a/\sqrt{n}, \epsilon/4}(\tilde{x}))\right] \right\} \\
&= \deg(S_{n/10, \theta, a/\sqrt{n}, \epsilon/4}(\tilde{x})) + \max \{\deg(A_{n, \theta, 2a}(x)), \deg(M_{n/10, \theta, \epsilon/4}(\tilde{x})) \} \\
&= 2 \cdot 41 \sqrt{\frac{n}{10} \ln(4/\epsilon)} + \max \left\{ 4 a \sqrt{n}, 41 \sqrt{\frac{n}{10} \ln(4/\epsilon)} \right\} \\
&= 2 \cdot 41 \sqrt{\frac{n}{10} \ln(4/\epsilon)} + \max \left\{ 4 \cdot (\sqrt{10} \sqrt{\ln (1/\epsilon)}) \cdot \sqrt{n}, 41 \sqrt{\frac{n}{10} \ln(4/\epsilon)} \right\} \\
&= 3 \cdot 41 \sqrt{\frac{n}{10} \ln(4/\epsilon)} \leq 41 \sqrt{n \ln(1/\epsilon)}.
\end{align*}

\paragraph{Time to compute $M_{n, \theta, \eps}$}

Computing $A_{n, \theta, 2a}$ can be done in $\poly{(n)}$ time as described in Lemma \ref{Apoly}, as can sampling $\tilde{x}$ from $x$. Given the three recursive $M_{n/10, \theta', \eps/4}$ polynomials, we can then compute $M_{n, \theta, \eps}$ in three multiplications. Each recursive polynomial has degree at most $d(n/10, \eps/4)$, and hence at most $\sum_{i=0}^{d(n/10,\eps/4)} {n \choose i}$ monomials. Since the time for these multiplications dominates the time for the recursive computations, the total time is $\tilde{O}(\sum_{i=0}^{d(n,\eps)} {n \choose i})$  using the fast Fourier transform\footnote{By replacing each variable with increasing powers of a single variable, we can reduce multivariate polynomial multiplication to single variable polynomial multiplication.}, as desired.

\paragraph{Correctness.} Now we prove that $M_{n,\theta,\eps}$ correctly simulates $\TH_\theta$ with probability at least $1-\eps$, on all possible inputs. We begin by citing two lemmas explaining our choice of the parameter $a$.

\begin{lemma}[Hoeffding's Inequality for Binomial Distributions (\cite{Hoeffding}~Theorem 1)] \label{hoeffding}
If $m$ independent random draws $x_1, \ldots, x_m \sim \{ 0,1 \}$ are made with $\Pr[x_i=1] = p$ for all $i$, then for any $k \leq mp$ we have
\[\Pr \left[ \sum_{i=1}^m x_i \leq k \right] \leq \exp\left(-\frac{2(mp-k)^2}{m}\right),\] where $\exp(x) = e^x$.
\end{lemma}

\begin{lemma} \label{lem}
If $x \in \{ 0,1 \}^n$ with $|x|/n = w$, and $\tilde{x} \in \{0,1\}^{n/10}$ is a vector each of whose entries is an independent and uniformly random entry of $x$, with $|\tilde{x}|/(n/10) = v$, then for every $\eps < 1/4$,
\[\Pr\left[v \leq w - a/\sqrt{n}\right] \leq \frac{\eps}{4},\]
where $a = \sqrt{10} \cdot \sqrt{\ln (1/\epsilon)}$.
\end{lemma}

\begin{proof} Each entry of $\tilde{x}$ is drawn from a binomial distribution with probability $w$ of giving a $1$. Hence, applying Lemma~\ref{hoeffding} with $p=w$, $m = n/10$, and $k = \frac{n}{10}(w - a/\sqrt{n}) = \frac{nw}{10} - \frac{a \sqrt{n}}{10}$ yields:
\[\Pr[v \leq w - a/\sqrt{n}] = \Pr\left[|\tilde{x}| \leq \frac{nw}{10} - \frac{a \sqrt{n}}{10}\right] \leq \exp\left(-2 \frac{\left( \frac{a \sqrt{n}}{10} \right)^2}{\frac{n}{10}}\right),\] which simplifies to $ \exp\left(- \frac{a^2}{5}\right) = \exp(-2 \ln(1/\epsilon)) = \epsilon^2 < \frac{\eps}{4}$.
\end{proof}

We now move on to the main proof of correctness, which proceeds by induction on $n$. By symmetry, we may assume we have an input vector $x \in \{0,1\}^n$ with $|x|/n \geq \theta$, and we want to show that $M_{n, \theta, \epsilon}(x)$ outputs $1$ with probability at least $1 - \epsilon$. We assume $\epsilon < 1/4$ so that we may apply Lemma \ref{lem}.

For notational convenience, define the intervals:
\[\alpha_0 = [\theta - a/\sqrt{n}, \theta], ~ \alpha_1 = [\theta, \theta + a/\sqrt{n}], ~ \beta = [\theta + a/\sqrt{n}, \theta + 2a/\sqrt{n}], ~ \gamma = [\theta + 2a/\sqrt{n}, 1].\]

Note that depending on the values of $\theta$ and $a$, some of these intervals may be empty; this is not a problem for our proof.

Let $w = |x|/n$. Let $\tilde{x}$ be the random ``subvector'' of $x$ selected in $M_{n, \theta, \epsilon}$ (recall we use the same $\tilde{x}$ in each of the three locations it appears in the definition of $M$). Let $v = |\tilde{x}|/(n/10)$. Our proof strategy is to consider different cases depending on the value of $w$. For each case, we show there are at most four events such that, if all events hold then $M_{n, \theta, \epsilon}$ outputs the correct answer, and each event does not hold with probability at most $\frac{\eps}{4}$. By the union bound, this implies that $M_{n, \theta, \epsilon}$ gives the correct answer with probability at least $1 - \epsilon$. The cases are as follows:

\begin{enumerate}
\item {\bf $w \in \alpha_1$ ($|x|/n$ is ``very close'' to $\theta$)}. By Lemma \ref{lem}, we know that with probability at least $1 - \frac{\epsilon}{4}$, we have $v \geq \theta - a/\sqrt{n}$. In other words, $v \in \alpha_0 \cup \alpha_1 \cup \beta \cup \gamma$. 
\begin{itemize}
\item {\bf $v \in \alpha_0 \cup \alpha_1$}, then with probability at least $1 - \frac{2\epsilon}{4}$, we have $S_{n/10, \theta, a/\sqrt{n}, \epsilon/4}(\tilde{x}) = 1$, by our inductive assumption that $S_{n/10, \theta, a/\sqrt{n}, \epsilon/4}$ is a probabilistic polynomial for $\NEAR_{\theta, a/\sqrt{n}}$ with error probability at most $\frac{2\epsilon}{4}$. In this case, $M_{n, \theta, \epsilon}(x) = A_{n, \theta, 2a}(x)$, which is $1$ by definition of $A$.
\item {\bf $v \in \beta \cup \gamma$}, then with probability at least $1 - \frac{2\epsilon}{4}$, we have $S_{n/10, \theta, a/\sqrt{n}, \epsilon/4}(\tilde{x}) = 0$, in which case $M_{n, \theta, \epsilon}(x) = M_{n/10, \theta, \epsilon/4}(\tilde{x})$. But, by the inductive hypothesis, this is $1$ with probability at least $1 - \frac{\epsilon}{4}$, since $v > \theta$ in this case.
\end{itemize}
Since we are in one of these two cases with probability $\geq 1 - \frac14 \epsilon$, and each gives the correct answer with probability $\geq 1 - \frac{3\epsilon}{4}$, the correct answer is given in this case with probability $\geq 1 - \epsilon$.

\item {\bf $w \in \beta$ ($|x|/n$ is ``close'' to $\theta$)}. In this case we have $w - \theta \leq 2a/\sqrt{n}$, therefore $A_{n, \theta, 2a}(x) = 1$. Hence, if $S_{n/10, \theta, a/\sqrt{n}, \epsilon/4}(\tilde{x}) = 1$ then $M_{n, \theta, \epsilon}(x)$ returns the correct answer. If $S_{n/10, \theta, a/\sqrt{n}, \epsilon/4}(\tilde{x}) = 0$, then we return $M_{n/10, \theta, \epsilon/4}(\tilde{x})$. By Lemma \ref{lem}, we have $v \geq \theta$ with probability at least $1 - \frac{\epsilon}{4}$, and in this case, $M_{n/10, \theta, \epsilon/4}(\tilde{x}) = 1$ with probability $\geq 1 - \frac{\epsilon}{4}$. Hence, $M$ returns the correct value with probability at least $1 - \frac{2\epsilon}{4}$, no matter what the value of $S_{n/10, \theta, a/\sqrt{n}, \epsilon/4}(y)$ happens to be.

\item {\bf $w \in \gamma$ ($|x|/n$ is ``far'' from $\theta$)}. By Lemma \ref{lem}, we have $v \in \beta \cup \gamma$ with probability at least $1 - \frac{\epsilon}{4}$. In this case, $v \geq \theta$, and so $M_{n/10, \theta, \epsilon/4}(\tilde{x}) = 1$ with probability $\geq 1 - \frac{\epsilon}{4}$. Moreover, since $v \notin \alpha_0 \cup \alpha_1$, it follows that $S_{n/10, \theta, a/\sqrt{n}, \epsilon/4}(\tilde{x}) = 0$ with probability $\geq 1 - \frac24 \epsilon$, in which case $M_{n, \theta, \epsilon}(x) = M_{n/10, \theta, \epsilon/4}(\tilde{x})$. Overall, $M_{n, \theta, \epsilon}(x) = M_{n/10, \theta, \epsilon/4}(\tilde{x}) = 1$ with probability $\geq 1 - \epsilon$.
\end{enumerate}

This completes the proof of correctness, and the proof of Theorem~\ref{MAJpoly}. 

\subsection{Symmetric Functions}

Recall that $f : \{ 0,1 \} \to \{ 0,1 \}^n$ is \emph{symmetric} if the value of $f(x)$ depends only on $|x|$, the Hamming weight of $x$. We now describe how to use the probabilistic polynomial for $\TH_\theta$ to derive a probabilistic polynomial for any symmetric function with the same degree as $\TH_{\theta}$:

\begin{reminder}{Theorem~\ref{SYMpoly}} Every symmetric function $f : \{ 0,1 \} \to \{ 0,1 \}^n$ on $n$ variables has a probabilistic polynomial of $O(\sqrt{n \log( 1 / \epsilon)})$ degree and error $\epsilon$. 
\end{reminder}

\begin{proof}
For any $0 \leq i \leq n$, let $f_i$ denote the value of $f(x)$ when $x$ has Hamming weight $i$. Define:$$A = \{ 0 < i \leq n \mid f_i = 1 \text{ and } f_{i-1} = 0 \},$$
$$B = \{ 0 < i \leq n \mid f_i = 0 \text{ and } f_{i-1} = 1 \}.$$

Then, $f$ can be written exactly as: \begin{align}\label{decomp}f(x) = f_0 + \sum_{i \in B} TH_{i/n}(x) - \sum_{j \in A} TH_{j/n}(x).\end{align}

We replace each $TH_\theta$ in \eqref{decomp} with a probabilistic polynomial of Theorem~\ref{MAJpoly} with error $\delta = \epsilon/2$. However, we make sure that in all of the different probabilistic polynomials for $TH_\theta$, we make the same choice for the sampled vector $\tilde{x}$ at each iteration. We can then apply the proof of Theorem~\ref{MAJpoly}, to see that every one of the $TH_\theta$ probabilistic polynomials will give the correct answer as long as $\left||x|/n - |\tilde{x}|/(n/10)\right| < a/\sqrt{n}$ at each of the $\log_{10}(n)$ layers of recursion (this is a property only of the sampling, and independent of $\theta$). Recall that the error parameter at the $i$th level of the recursion is $\frac{1}{4^i} \delta$. Hence, by the union bound, the error probability of the entire probabilistic polynomial is at most

$$\delta + \frac{1}{4}\delta + \frac{1}{16}\delta + \cdots + \frac{1}{4^{\log_{10}}(n)} \delta < \frac{1}{1-1/4}\delta < \epsilon,$$
as desired.
\end{proof}

\section{Closest Pair in Hamming Space, and Batch Nearest Neighbor}

We first give a connection between the time complexity of closest pair problems in metric spaces on the hypercube and the existence of certain probabilistic polynomials. Let $M$ be a metric on $\{0,1\}^d$. We define the {\sc Bichromatic $M$-Metric Closest Pair} problem to be: given an integer $k$ and a collection of ``red'' and ``blue'' vectors in $\{0,1\}^d$, determine if there is a pair of red and blue vectors with distance at most $k$ under metric $M$. This problem arises frequently in algorithms on a metric space $M$. In what follows, we shall assume that the metric $M$ can be computed on two points of $d$ dimensions in time $\poly(d)$. Define the Boolean function \begin{align*}
\text{$M$-dist}_k(x_{1,1},\ldots,x_{1,d},\ldots,x_{s,1},\ldots,x_{s,d},y_{1,1},\ldots,y_{1,d},\ldots,y_{s,1},\ldots,y_{s,d})\\ := \bigvee_{i,j = 1,\ldots,s} \text{\bf [}M(x_{i,1},\ldots,x_{i,d},y_{j,1},\ldots,y_{j,d}) \leq k\text{\bf ]}.\end{align*} That is, $\text{$M$-dist}_k$ takes two collections of $s$ vectors as input, and outputs $1$ if and only if there is a pair of vectors (one from each collection) that have distance at most $k$ under metric $M$. For example, the $\text{Hamming-dist}_k$ function tests if there is a pair of vectors with Hamming distance at most $k$.

We observe that sparse probabilistic polynomials for computing $\text{$M$-dist}_k$ imply subquadratic time algorithms for finding close bichromatic pairs in metric $M$. 

\begin{theorem} \label{probpoly2Bichromatic} Suppose for all $k$, $d$, and $n$, there is an $s = s(d,n)$ such that $\text{$M$-dist}_k$ on $2sd$ variables has a probabilistic polynomial with at most $n^{0.17}$ monomials and error at most $1/3$, where each sample can be produced in $\tilde{O}(n^2/s^2)$ time. Then {\sc Bichromatic $M$-Metric Closest Pair} on $n$ vectors in $d$ dimensions can be solved in $\tilde{O}(n^2/s^2 + s^2\cdot \poly(d))$ randomized time.
\end{theorem}

\begin{proof} We have an integer $k$ and sets $R, B \subseteq \{0,1\}^d$ such that $|R|=|B|=n$, and wish to determine if there is a $u \in R$ and $v \in B$ such that $M(u,v) \leq k$. First, partition both $R$ and $B$ into $\lceil n/s \rceil$ groups, with at most $s$ vectors in each group. By assumption, for all $k$, there is a probabilistic polynomial for $\text{$M$-dist}_k$ with $2sd$ variables, $n^{0.17}$ monomials, and error at most $1/3$. Let $p$ be a polynomial sampled from this distribution. Our idea is to efficiently evaluate $p$ on all $O(n^2/s^2)$ pairs of groups from $R$ and $B$, by feeding as input to $p$ all $s$ vectors $x_i$ from a group of $R$ and all $s$ vectors $y_i$ from a group of $B$. 

Since the number of monomials $m \leq n^{0.17}$, we can apply Lemma~\ref{polyeval}, evaluating $p$ on all pairs of groups in time $\tilde{O}(n^2/s^2)$. For each pair of groups from $R$ and $B$, this evaluation determines if the pair of groups contain a bichromatic pair of distance at most $k$, with probability at least $2/3$. 

To obtain a high probability answer, sample $\ell=10\log n$ polynomials $p_1,\ldots,p_{\ell}$ for $\text{$M$-dist}_k$ independently from the distribution, in $\tilde{O}(n^2/s^2)$ time (by assumption). Evaluate each $p_i$ on all pairs of groups from $R$ and $B$ in $\tilde{O}(n^2/s^2)$ time by the above paragraph. Compute the majority value of $p_1,\ldots,p_{\ell}$ on all pairs of groups, again in $\tilde{O}(n^2/s^2)$ time. By Chernoff-Hoeffding bounds, the majority value reported for a pair of groups is correct with probability at least $1-n^{-3}$. Therefore with probability at least $1-n^{-1}$, we correctly determine for all pairs of groups from $R$ and $B$ whether the pair contains a bichromatic pair of vectors with distance at most $k$. 

Given a pair of groups $R'$ and $B'$ which are reported to contain a bichromatic pair of close vectors,we can simply brute force to find the closest pair in $A'$ and $B'$ in $s^2 \cdot \poly(d)$ time. (In principle, we could also perform a recursive call, but this doesn't asymptotically help us in our applications.)
\end{proof}

Next, we construct a probabilistic polynomial for the  $\text{Hamming-dist}_k$ function, using the MAJORITY construction of Theorem~\ref{MAJpoly}. 

\begin{theorem} \label{Hammingprobpoly} There is a $e \geq 1$ such that for sufficiently large $s$ and $d > e^2 \log s$, the $\text{Hamming-dist}_k$ function on $2sd$ variables has a probabilistic polynomial of degree $O(\sqrt{d \log s})$, error at most $1/3$, and at most $O(s^4 \cdot {2d \choose O(\sqrt{d \log s})})$ monomials over $\F_2$. Moreover, we can sample from the probabilistic polynomial distribution in time polynomial in the number of monomials.
\end{theorem}

A similar result holds for $\Z$, as well as any field, with minor modifications. (For fields of characteristic $p$, the degree increases by a factor of $p-1$.)

\begin{proof} 
Let $e \geq 1$ be large enough that there is a probabilistic polynomial ${\cal D}_d$ of degree $e\sqrt{d \log(1/\eps)}$ for the threshold function $\TH_{(k+1)/d}$ on $d$ inputs, from Theorem~\ref{MAJpoly}. We construct a probabilistic polynomial ${\cal H}$ for $\text{Haming-dist}_k$ over $\F_2$, as follows: 

Set $\eps = 1/s^3$, and sample $p \sim {\cal D}_d$ with error $\eps$. Let $x_1, y_1,\ldots,x_s,y_s$ be blocks of $d$ Boolean variables, with the $j$th variable of $x_i$ denoted by $x_{i,j}$. Choose two uniform random subsets $R_1,R_2 \subseteq [s]^2$, and form \[q(x_1,y_1,\ldots,x_s,y_s) 
:= 1+\prod_{k=1}^2\left(1+\sum_{(i,j) \in R_k} (1+p(x_{i,1}+y_{j,1},\ldots,x_{i,d}+y_{j,d}))\right).\] First, note that since $\eps = 1/s^3$, all $2s^2$ occurrences of the polynomial $p$ in $q$ output the correct answer with probability at least $1-2/s$. Let us suppose this event occurs.

If there are $x_i$ and $y_i$ with Hamming distance at most $k$, then $p(x_{i,1}+y_{j,1},\ldots,x_{i,d}+y_{j,d})) = 0$ (recall the summation is modulo $2$). Hence the probability that the sum of $(1+p)$'s in $R_1$ is odd is $1/2$. The same is true of $R_2$ independently. Therefore the product of the two sums in the expression for $q$ is $0$ with probability $3/4$, so $q$ outputs $1$ with probability $3/4$. On the other hand, if every $x_i$ and $y_i$ has Hamming distance at least $k$, then $1+p(x_{i,1}+y_{j,1},\ldots,x_{i,d}+y_{j,d}) = 0$ for all $(i,j) \in R_1 \cup R_2$. Therefore the product of the two sums (over $R_1$ and $R_2$) in $q$ is $1$,  hence $q$ outputs $0$ in this case. This shows that $q$ agrees with $\text{Hamming-dist}_k$ on any given input, with probability at least $3/4 - 2/s > 2/3$.

Now we prove the monomial bound. Since we are only evaluating $q$ on $0/1$ points, we may assume $q$ is multilinear, and remove all higher powers of the variables. Assuming $d > e\sqrt{d \log s}$, i.e. \begin{align} \label{requirement} d > e^2 \log s,\end{align} the number of distinct monomials in the multilinear $q$ is at most $O(s^4\cdot {2d \choose e\sqrt{d \log(s)}})$.
\end{proof}

Putting it all together, we obtain a faster algorithm for {\sc Bichromatic Hamming Closest Pair}:

\begin{theorem} \label{BichromaticHD} For $n$ vectors of dimension $d = c(n) \log n$, {\sc Bichromatic Hamming Closest Pair} can be solved in $n^{2-1/O(c(n) \log^2 c(n))}$ time by a randomized algorithm that is correct with high probability. 
\end{theorem}

\begin{proof} Let $d = c \log n$ in the following, with the implicit understanding that $c$ is a function of $n$. We apply the reduction of Theorem~\ref{probpoly2Bichromatic} and the probabilistic polynomial for the $\text{Hamming-dist}_k$ of Theorem~\ref{Hammingprobpoly}.

The reduction of Theorem~\ref{probpoly2Bichromatic} requires that the number of monomials in our probabilistic polynomial is at most $n^{0.17}$, while the monomial bound for $\text{Hamming-dist}_k$  from Theorem~\ref{Hammingprobpoly} is $m = O(s^2 \cdot {2d \choose e\sqrt{d \log s}})$ for some universal constant $a$, provided that $d > a^2 \log s$. Therefore our primary task is to maximize the value of $s$ such that $m \leq n^{0.17}$. This will minimize the final running time of $\tilde{O}(n^2/s^2)$. 
With hindsight, let us guess $s = n^{1/(u c \log^2 c)}$ for a constant $u$, and focus on the large binomial in the monomial estimate $m$. Then, \[{2d \choose a\sqrt{d \cdot \log s}} 
= {2c \log n \choose a\sqrt{(c \log n) \cdot (\log n)/(u c \log^2 c)}} 
= {2c \log n \choose a\sqrt{(\log^2 n)/(u \log^2 c)}} 
=  {2c \log n \choose a \log n/(\sqrt{u} \log c)}.\] 
For notational convenience, let $\delta = a/(\sqrt{u} \log c)$.
By Stirling's inequality, we have 
\[{2c \log n \choose \delta \log n}
\leq \left(\frac{2ce}{\delta}\right)^{\delta \log n}
= n^{\delta \log(\frac{2ce}{\delta})}.\] Plugging $\delta = a/(\sqrt{u} \log c)$ back into the exponent, we find 
\begin{align} \label{exponent1}
\delta \log\left(\frac{2ce}{\delta}\right) &= \frac{a \log(\frac{2ce\sqrt{u} \log c}{a})}{\sqrt{u} \log c}.\end{align} The quantity \eqref{exponent1} can be made arbitrarily small, by setting $u$ sufficiently large. In that case, the number of monomials $m \leq s^2 n^{\delta \log(\frac{2ce}{\delta})}$ can be made less than $n^{0.1}$. Finally, note that $a^2 \log s = a^2 (\log n)/(u c \log^2 c) < c \log n = d$, so \eqref{requirement} holds and the reduction of Theorem~\ref{probpoly2Bichromatic} applies. This completes the proof. 
\end{proof}

Observe that the probabilistic polynomials of degree $\sqrt{n \log (1/\eps)} \polylog{n}$ from prior work~\cite{Srinivasan13} would be insufficient for Theorem~\ref{BichromaticHD}. The extra degree increase would include an extra $\polylog{n}$ factor in expression \eqref{exponent1}, and hence no constant choice of $u$ would be sufficiently large.

Now we show how to solve {\sc Batch Hamming Nearest Neighbor} (BHNN). In the following theorem, we assume for all pairs of vectors in our instance that the maximum metric distance is at most some value $MAX$. (For the Hamming distance, $MAX \leq d$.) We reduce the batch nearest neighbor query problem to the bichromatic close pair problem:

\begin{theorem} Let $E^d$ be some $d$-dimensional domain supporting a metric space $M$. \label{BIC2NN} If the {\sc Bichromatic $M$-Metric Closest Pair} on $n$ vectors in $E^d$ can be solved in $T(n,d)$ time, then {\sc Batch $M$-Metric Nearest Neighbors} on $n$ vectors in $E^d$ can be solved in $O(n \cdot T(\sqrt{n}, d)\cdot MAX)$ time.
\end{theorem}

\begin{proof} We give an oracle reduction similar to previous work~\cite{AbboudWY15}. Initialize an table $T$ of size $n$, with the maximum metric value $v$ in each entry. Given $n$ database vectors $D$ and $n$ query vectors $Q$, color $D$ red and $Q$ blue. Break $D$ into $\lceil n/s \rceil$ groups of size at most $s$, and do the same for the set $Q$. For each pair $(R',B') \subset (D \times Q)$ of groups, and for each $k = MAX-1,\ldots,1,0$, we initialize $D_k := D$, $Q_k := Q$, and call {\sc Bichromatic $M$-Metric Closest Pair} on $(R',B') \subset (D_k \times Q_k)$ with integer $k$. While we continue to find a pair $(x_i,y_j) \in (R' \times B')$ with $M(x_i,y_j) \leq k$, set $T[i] := k$ and remove $y_j$ from $Q_k$ and $B'$. (With a few more recursive calls, we could also find an explicit vector $y_j$ such that $M(x_i,y_j) \leq k$.) 

Now for each call that finds a close bichromatic pair, we remove a vector from $Q_k$; we do this at most $MAX$ times for each vector, so there can be at most $MAX \cdot n$ such calls. For each pair of groups, there are $MAX$ oracle calls that find no bichromatic pair. Therefore the total running time is $O((n+n^2/s^2) \cdot T(s, d)\cdot MAX)$. Setting $s=\sqrt{n}$ to balance the terms, the running time is $O(n \cdot T(\sqrt{n}, d)\cdot MAX)$. 
\end{proof}

The following is immediate from Theorem~\ref{BIC2NN} and Theorem~\ref{BichromaticHD}:

\begin{reminder}{Theorem~\ref{HammingNN}} For $n$ vectors of dimension $d = c(n) \log n$, {\sc Batch Hamming Nearest Neighbors} can be solved in $n^{2-1/O(c(n) \log^2 c(n))}$ time by a randomized algorithm, whp. 
\end{reminder}

\subsection{Some Applications}

Recall that the $\ell_1$ norm of two vectors $x$ and $y$ is $\sum_{i}|x_i - y_i|$. We can solve {\sc Batch $\ell_1$ Nearest Neighbors} on vectors with small integer entries by a simple reduction to {\sc Batch Hamming Nearest Neighbors}, (which is probably folklore):

\begin{theorem}\label{L1NN} For $n$ vectors of dimension $d = c(n) \log n$ in $\{0, 1, \ldots, m\}^d$, {\sc Batch $L_1$ Nearest Neighbors} can be solved in $n^{2-1/O(m c(n) \log^2 (m c(n)))}$ time by a randomized algorithm, whp. 
\end{theorem}

\begin{proof}
Notice that for any $x,y \in \{0, \ldots, m \}$, the Hamming distance of their unary representations, written as $m$-dimensional vectors, is equal to $|x-y|$. Hence, for $x \in \{0, \ldots, m\}^d$, we can transform it into a vector $x' \in \{ 0, 1\}^{md}$ by setting $(x'_{m(i-1)+1}, x'_{m(i-1)+2}, \ldots, x'_{m(i-1)+m})$ equal to the unary representation of $x_i$, for $1 \leq i \leq d$. It is then equivalent to solve the Hamming nearest neighbors problem on these $md$-dimensional vectors.
\end{proof}

It is also easy to extend Theorem~\ref{HammingNN} for vectors over $O(1)$-sized alphabets using equidistant binary codes~(\cite{MinKerui}, Section 5.1). This is useful for applications in biology, such as finding similar DNA sequences. The above algorithms also apply to computing maximum inner products:

\begin{theorem} \label{BichromaticMIP}
The {\sc Bichromatic Minimum Inner Product} (and {\sc Maximum}) problem with $n$ red and blue Boolean vectors in $c \log n$ dimensions can be solved in $n^{2-1/O(c \log^2 c)}$ randomized time. 
\end{theorem}

\begin{proof} Recall that Theorem \ref{SETH} gives a reduction from {\sc Bichromatic Minimum Inner Product} to {\sc Bichromatic Hamming Furthest Pair}, \emph{and} shows that {\sc Bichromatic Hamming Furthest Pair} is equivalent to {\sc Bichromatic Hamming Closest Pair}. The same reduction shows that {\sc Bichromatic Maximum Inner Product} reduces to the closest pair version. Hence Theorem~\ref{HammingNN} applies, to both minimum and maximum inner products.
\end{proof}

As a consequence, we can answer a batch of $n$ minimum inner product queries on a database of size $n$ with the same time estimate, applying a reduction analogous to that of Theorem~\ref{BIC2NN}. From there, Theorem~\ref{BichromaticMIP} can be extended to other important similarity measures, such as finding a pair of sets $A, B$ with maximum \emph{Jaccard coefficient}, defined as $\frac{|A \cap B|}{|A \cup B|}$ ~\cite{Broder97}.

\begin{corollary}
Given $n$ red and blue sets in $\{0,1\}^{c \log n}$, we can find the pair of red and blue sets with maximum Jaccard coefficient in $n^{2-1/O(c \log^2 c)}$ randomized time. 
\end{corollary}

\begin{proof} Let $S$ be a given collection of red and blue sets over $[d]$. We construe the sets in $S$ as vectors, in the natural way. For all possible values $d_1, d_2 = 1,\ldots,d$, we will construct an instance of {\sc Bichromatic Maximum Inner Product} $S'_{d_1,d_2}$, and take the best pair found, appealing to Theorem~\ref{BichromaticMIP}.

As in the proof of Theorem~\ref{SETH}, we ``filter'' sets based on their cardinalities. In the instance $S'_{d_1,d_2}$ of {\sc Bichromatic Maximum Inner Product}, we only include red sets with cardinality exactly $d_1$, and blue sets with cardinality exactly $d_2$. For sets $R,B$, we have \begin{align}\label{jaccard}\frac{|R \cap B|}{|R \cup B|} = \frac{|R \cap B|}{d_1 + d_2 - |R \cap B|}.\end{align} Suppose that we choose a red set $R$ and blue set $B$ that maximize $|R \cap B|$. This choice simultaneously maximizes the numerator and minimizes the denominator of \eqref{jaccard}, producing the sets $R$ and $B$ with maximum Jaccard coefficient over the red sets with cardinality $d_1$ and blue sets with cardinality $d_2$. Finding the maximum pair $R$ and $B$ over each choice of $d_1, d_2$, we will find the overall $R$ and $B$ with maximum Jaccard coefficient.
\end{proof}

\subsection{Closest Pair in Hamming Space is Hard}\label{appendix-SETH}

The \emph{Strong Exponential Time Hypothesis} (SETH) states that there is no universal $\delta < 1$ such that for all $c$, CNF-SAT with $n$ variables and $cn$ clauses can be solved in $O(2^{\delta n})$ time. 

\begin{reminder}{Theorem~\ref{SETH}} Suppose there is $\eps > 0$ such that for all constant $c$, {\sc Bichromatic Hamming Closest Pair} can be solved in $2^{o(d)} \cdot n^{2-\eps}$ time on a set of $n$ points in $\{0,1\}^{c \log n}$. Then SETH is false.
\end{reminder}

\begin{proof} The proof is a reduction from the {\sc Orthogonal Vectors} problem with $n$ vectors $S \subset \{0,1\}^d$: are there $u, v \in S$ such that $\langle u,v\rangle = 0$? It is well-known that $2^{o(d)} \cdot n^{2-\eps}$ time would refute SETH~\cite{Williams05}. Indeed, we show that {\sc Bichromatic Minimum Inner Product} (finding a pair of vectors with minimum inner product, not just inner product zero) reduces to {\sc Bichromatic Hamming Closest Pair}, as well as the version for maximum inner product.

First, we observe that {\sc Bichromatic Hamming Closest Pair} is equivalent to {\sc Bichromatic Hamming Furthest Pair}: let $\overline{v}$ be the complement of $v$ (the vector obtained by flipping all the bits of $v$). Then the Hamming distance of $u$ and $v$ is $H(u,v) = d-H(u,\overline{v})$. Thus by flipping all the bits in the components of the blue vectors, we can reduce from the closest pair problem to furthest pair, and vice versa.

Now we reduce {\sc Orthogonal Vectors} to {\sc Bichromatic Hamming Furthest Pair}. Our {\sc Orthogonal Vectors} instance has red vectors $S_r$ and blue vectors $S_b$, and we wish to find $u \in S_r$ and $v \in S_b$ such that $\langle u,v\rangle = 0$. 

For every $d^2$ possible choice of $I, J = 1,\ldots,d$, construct the subset $S_{r,I}$ of vectors in $S_r$ with exactly $I$ ones,  and construct the subset $S_{b,J}$ of vectors in $S_b$ with exactly $J$ ones. We will look for an orthogonal pair among $S_{r,I}$ and $S_{b,J}$ for all such $I,J$ separately.

Recall that Hamming distance of two vectors equals the $\ell_2^2$ norm distance, in $\{0,1\}^d$. The $\ell_2^2$ norm of $u$ and $v$ is \[||u-v||_2^2 = ||u||_2 + ||v||_2 - 2 \langle u,v\rangle.\] However, in $S_{r,I}$ all vectors have the same norm, and all vectors in $S_{b,J}$ have the same norm. Therefore, finding a red-blue pair $u \in S_{r,I}$ and $v \in S_{b,J}$ with minimum inner product is equivalent to finding a pair in $S_r \times S_b$ with smallest Hamming distance. (Similarly, maximum inner product is equivalent to Hamming closest pair.)

The reduction only requires $O(d^2)$ calls to {\sc Bichromatic Hamming Furthest Pair}, with no changes to the dimension $d$ nor the number of vectors $n$.
\end{proof}

\section{Conclusion}

There are many interesting further directions. Here are some general questions about the future of this approach for nearest neighbor problems:
\begin{itemize}
\item Could a similar approach solve the closest pair problem for edit distance in $\{0,1\}^d$? This is a natural next step. Reductions from edit distance to Hamming distance are known~\cite{Bar-Yossef04} but they yield large approximation factors; we think exact solutions should be possible. The main difficulty is that the circuit complexity (and probabilistic polynomial degree) of edit distance seems much higher than that of Hamming distance: Hamming distance can be seen as a ``threshold of XORs'', but the best complexity upper bound for edit distance seems to be ${\sf NLOGSPACE}$ 
.

\item We can solve the off-line ``closest pair'' version of several data structure problems, by reducing them to problems of evaluating polynomials, and applying matrix multiplication. Is there any way to obtain better \emph{data structures} using this algebraic approach? Of course there are limitations on such data structures, there are also gaps between known data structures and known lower bounds.

\item It feels strange to embed multivariate polynomial evaluations into a matrix multiplication, when it is known that evaluating univariate polynomials on many points can be done even faster than known matrix multiplication algorithms (using FFTs). Perhaps we can apply other algebraic tools (such as Kedlaya and Umans' multivariate polynomial evaluation algorithms~\cite{Umans08,KU11}) directly to these problems.

\item Recently, Timothy Chan \cite{Chan15} gave an algorithm for computing dominances among $n$ vectors in $\R^{c \log n}$, which has a running time that is very similar to ours: $n^{2-1/O(c \log^2 c)}$ time. Is this a coincidence?

\end{itemize}

\section{Acknowledgements}

We thank an anonymous FOCS reviewer for pointing out that our probabilistic polynomial for general symmetric functions can achieve an $O(\sqrt{n \log(1/\eps)})$ degree bound as well.

\bibliographystyle{alpha}
\bibliography{papers}

\end{document}